\documentclass{easychair}

\usepackage{xcolor}
\usepackage{float}
\usepackage{amsmath}
\usepackage{amssymb}
\usepackage {mathpartir}
\usepackage{stmaryrd}
\usepackage{tikzfig}
\usepackage{wrapfig}
\usepackage{url}
\tikzstyle{every picture}=[baseline=-0.25em]
\tikzstyle{dotpic}=[scale=0.6]
\tikzstyle{diredges}=[every to/.style={diredge}]
\tikzstyle{dot graph}=[shorten <=-0.1mm,shorten >=-0.1mm,scale=0.6]
\tikzstyle{digraph}=[-latex]
\tikzstyle{plot point}=[circle,fill=black,minimum width=2mm,inner sep=0]
\tikzstyle{string graph}=[scale=0.6]
\tikzstyle{sg diredge}=[-stealth]
\tikzstyle{rewrite edge}=[-open triangle 45]
\tikzstyle{sg bold diredge}=[-stealth,thick,shorten >=-1pt]
\tikzstyle{sg vertex}=[circle,minimum width=2.2mm,fill=white,draw=black,inner sep=0mm]
\tikzstyle{labelled sg vertex}=[circle,minimum width=7mm,fill=white,draw=black,inner sep=0mm]
\tikzstyle{sg grey vertex}=[sg vertex,fill=gray!30!white]
\tikzstyle{sg black vertex}=[sg vertex,fill=black]
\tikzstyle{sg bold vertex}=[circle,minimum width=2.2mm,fill=white,draw=black,very thick,inner sep=0mm]
\tikzstyle{sg wire vertex}=[circle,minimum width=1mm,fill=black,inner sep=0mm]
\tikzstyle{cloud vertex}=[fill=white, draw=black, inner sep=2 mm, shape=cloud, aspect=1.5]
\tikzstyle{tick vertex}=[rectangle,fill=black,minimum height=1mm,minimum width=2.5mm,inner sep=0mm]

\tikzstyle{braceedge}=[decorate,decoration={brace,amplitude=2mm,raise=-1mm}]
\tikzstyle{small braceedge}=[decorate,decoration={brace,amplitude=1mm,raise=-1mm}]
\tikzstyle{left hook arrow}=[left hook-latex]
\tikzstyle{right hook arrow}=[right hook-latex]

\tikzstyle{bbox edge}=[draw=blue]
\tikzstyle{bbox include}=[->,draw=blue]
\tikzstyle{bbox corner}=[inner sep=0pt,rectangle,fill=blue,draw=blue,minimum width=1.5mm,minimum height=1.5mm]

\tikzstyle{west wire label}=[font=\footnotesize,anchor=west,inner sep=1pt,xshift=-3pt]
\tikzstyle{east wire label}=[font=\footnotesize,anchor=east,inner sep=1pt,xshift=3pt]

\tikzstyle{dot}=[inner sep=0.7mm,minimum width=0pt,minimum height=0pt,fill=black,draw=black,shape=circle]
\tikzstyle{white dot}=[dot,fill=white]
\tikzstyle{alt white dot}=[white dot,label={[xshift=2.9mm,yshift=-0.1mm]left:$\cdot$}]
\tikzstyle{gray dot}=[dot,fill=gray!50]
\tikzstyle{box vertex}=[draw=black,rectangle]
\tikzstyle{whitebg}=[fill=white,inner sep=2pt]
\tikzstyle{graph state vertex}=[sg vertex,fill=black]

\tikzstyle{square box}=[rectangle,fill=white,draw=black,minimum height=6mm,minimum width=6mm]
\tikzstyle{square gray box}=[rectangle,fill=gray!30,draw=black,minimum height=6mm,minimum width=6mm]
\tikzstyle{point}=[regular polygon,regular polygon sides=3,draw=black,scale=0.75,inner sep=-0.5pt,minimum width=7mm,fill=white]
\tikzstyle{copoint}=[point,regular polygon rotate=180,fill=white]
\tikzstyle{gray point}=[point,fill=gray!40!white]
\tikzstyle{gray copoint}=[copoint,fill=gray!40!white]

\tikzstyle{open graph}=[baseline=-0.25em]
\tikzstyle{greybg}=[background rectangle/.style={fill=black!5,draw=black!30,rounded corners=1ex}, show background rectangle]

\tikzstyle{edge point}=[circle,minimum width=1mm,fill=black,inner sep=0mm]
\tikzstyle{vertex point}=[circle,minimum width=2.2mm,fill=white,draw=black,inner sep=0mm]
\tikzstyle{gray vertex point}=[circle,minimum width=2.2mm,fill=gray!30!white,draw=black,inner sep=0mm]
\tikzstyle{edge label}=[inner sep=2pt, font=\small]
\tikzstyle{on edge label}=[fill=white, font=\footnotesize, inner sep=1 pt]

\newcommand{\edgearrow}{{\arrow[black]{>}}}
\newcommand{\edgetick}{{\arrow[black,scale=0.7,very thick]{|}}}
\tikzstyle{diredge}=[->]
\tikzstyle{medium diredge}=[->]
\tikzstyle{short diredge}=[->]
\tikzstyle{halfedge}=[-)]
\tikzstyle{other halfedge}=[(-]
\tikzstyle{freeedge}=[(-)]
\tikzstyle{white edge}=[line width=5pt,white]
\tikzstyle{tick}=[postaction=decorate,decoration={markings, mark=at position 0.5 with \edgetick}]
\tikzstyle{small map edge}=[|-latex, gray!60!blue, shorten <=0.9mm, shorten >=0.5mm]
\tikzstyle{thick dashed edge}=[very thick,dashed,gray!40]
\tikzstyle{map edge}=[|-latex,very thick, gray!40, shorten <=1mm, shorten >=0.5mm]
\tikzstyle{tickedge}=[postaction=decorate,
  decoration={markings, mark=at position 0.5 with \edgetick}]
\tikzstyle{dirtickedge}=[postaction=decorate,
  decoration={markings, mark=at position 0.5 with \edgetick},
  decoration={markings, mark=at position 0.85 with \edgearrow}]
\tikzstyle{dirdoubletickedge}=[postaction=decorate,
  decoration={markings, mark=at position 0.4 with \edgetick},
  decoration={markings, mark=at position 0.6 with \edgetick},
  decoration={markings, mark=at position 0.85 with \edgearrow}]

\tikzstyle{arrs}=[-latex,font=\small,auto]
\tikzstyle{arrow plain}=[arrs]
\tikzstyle{arrow dashed}=[dashed,arrs]
\tikzstyle{arrow bold}=[very thick,arrs]
\tikzstyle{arrow hide}=[draw=white!0,-]
\tikzstyle{arrow reverse}=[latex-]
\tikzstyle{cdnode}=[]

\tikzstyle{cnot}=[fill=white,shape=circle,inner sep=-1.4pt]
\tikzstyle{bang box}=[draw=black,dashed,minimum height=12mm,minimum width=12mm,fill=gray!20]
\tikzstyle{wire label}=[font=\footnotesize, auto]

\newcommand{\cmdrewritesto}{\tikz[baseline=-0.25em] { \draw [-open triangle 45, line width=0.2pt] (0,0) -- (0.5,0); }\,}

\DeclareMathOperator{\rewritesto}{\cmdrewritesto}

\newtheorem{definition}{Definition}
\newtheorem{theorem}{Theorem}

\newcommand{\subtype}{\mbox{\texttt <:}}

\newcommand{\seml}{\mbox{$\llbracket$}}
\newcommand{\semr}{\mbox{$\rrbracket$}}
\newcommand{\sem}[1]{\seml{} #1 \semr{}}

\definecolor{light-gray}{gray}{0.75}
\newcommand{\ncbox}[1]{\fcolorbox{light-gray}{light-gray}{#1}} 	
\newcommand{\cbox}[1]{\small{\ncbox{#1}}}

\title{Towards Automated Proof Strategy Generalisation} 
\titlerunning{Towards Automated Proof Strategy Generalisation} 

\authorrunning{G. Grov \& E. Maclean}
\author{Gudmund Grov \\
School of Mathematical and Computer Sciences, 
Heriot-Watt University, Edinburgh, UK \\
\url{G.Grov@hw.ac.uk}
\and
Ewen Maclean \\
School of Informatics, University of Edinburgh, UK \\
\url{E.Maclean@ed.ac.uk}
}

\begin{document}
 
 \maketitle
 
\begin{abstract}
 The ability to automatically generalise (interactive) proofs and use
  such generalisations to discharge related conjectures is a very hard
  problem which remains unsolved; this paper shows how we hope to make a start on solving
  this problem. We develop a notion of \emph{goal types} to capture key properties of goals, which enables abstractions over the specific order and number of sub-goals arising when composing tactics. We show
  that the goal types form a lattice, and utilise this property in the techniques we develop to automatically generalise proof strategies in order to reuse it for proofs of related conjectures.
We illustrate our approach with an example.
\end{abstract}


\section{Introduction}\label{sec:intro}

When verifying large systems one often ends up applying the same proof strategy many times -- albeit with small variations. An 
expert user/developer of a theorem proving system would often implement common proof patterns as a so-called \emph{tactics}, and use this to automatically discharge ``similar" conjectures. However, other users often need to manually prove each conjecture. Our ultimate goal is to automate the process of generalising a proof
(possibly a few proofs) into a sufficiently generic proof strategy capable of proving ``similar" conjectures. In this paper we make a small step towards this goal by developing a suitable 
representation with necessary strong formal properties, and give two generic methods which utilises this representation to generalise a proof.

Whilst the manual repetition of similar proofs have been observed across different formal methods, for example Event-B, B and VDM (see \cite{AI4FM:avocs09}), we will focus on a subset of \emph{separation logic} \cite{SepLogic02}, used to reason about pointer-based programs\footnote{However, note that we believe that our approach is still generic across different formal methods.}. In the subset, there are
two binary operations $*$ and $\wedge$ and a predicate $pure$, with
the following axioms:
$$
\begin{array}{rrlr}
(A * B) * C &\Leftrightarrow&A * (B * C) & \qquad (\textsf{ax1}) \\
pure(B) \;\to\; (A \wedge B) * C&\Leftrightarrow& (A * C) \wedge B
 & \qquad (\textsf{ax2}) \\
\end{array}
$$
These axioms pertain specifically to separation logic, and allow {\em
  pure}/functional content to be expressed apart from {\em shape} 
content, used to describe resources. 
Now, consider the conjecture:
\begin{eqnarray}
&p: pure(e),h: c * ((f * (d * b) \wedge e) \wedge e) * a
\vdash & ((c*f)*(d \wedge e)) * ((b \wedge e) * a)\label{eq:g1}
\end{eqnarray}
which demonstrate the typical form of a
goal resulting from proving properties about heap structures, which
involve some resource content and some functional content. For
example, one could view the $a,b,c,d,f$ as propositions about space on a heap, with $e$ containing some functional information -- for example about order. 

\begin{figure}
\begin{tabular}{ll}
\begin{minipage}[c]{0.6\textwidth}
\begin{tabbing}
   \=\textbf{\textsf{lemma}} \textbf{\textsf{assumes}} \textsf{p}: $pure (e)$ \\
\> $\quad$ \=\textbf{and} \textsf{h}: $c * ((f * (d * b) \wedge e) \wedge e) * a$ \\
\>\>\textbf{\textsf{sh}}\=\textbf{\textsf{ows}} \cbox{$g_1:$ $((c*f)*(d \wedge e)) * ((b \wedge e) * a)$} $\qquad$
\\
\>\>\>\textbf{\textsf{apply}} (subst ax1) ~\= 
\cbox{[$g_2: ((c * (f \wedge e) * d \wedge e) * b) * a$]} \\
\>\>\>\textbf{\textsf{apply}} (subst ax1)\> 
\cbox{[$g_3: (c * ((f \wedge e) * d \wedge e) * b) * a$]}\\
\>\>\>\textbf{\textsf{apply}} (subst ax1)\> 
\cbox{[$g_4: (c * (f \wedge e) * (d \wedge e) * b) * a$]}\\
\>\>\>\textbf{\textsf{apply}} (subst ax1)\> 
\cbox{[$g_5: c * ((f \wedge e) * (d \wedge e) * b) * a$]}\\
\>\>\>\textbf{\textsf{apply}} (subst ax2)\> 
\cbox{[$g_6: pure (e),~ g_7: c * ((f * (d \wedge e) * b) \wedge e) * a$]}\\
\>\>\>\textbf{\textsf{apply}} (rule p)\> 
\cbox{[$g_7: c * ((f * (d \wedge e) * b) \wedge e) * a$]}\\
\>\>\>\textbf{\textsf{apply}} (subst ax2)\> 
\cbox{[$g_8: pure (e),~ g_{9}:c * ((f * (d * b) \wedge e) \wedge e) * a$]}\\
\>\>\>\textbf{\textsf{apply}} (rule p)\> 
\cbox{[$g_{9}:c * ((f * (d * b) \wedge e) \wedge e) * a$]}\\
\>\>\>\textbf{\textsf{apply}} (rule h)\> 
\cbox{[$\quad$]}\\
\>\>\>\textbf{\textsf{done}}
\end{tabbing}
\end{minipage}
&
\begin{minipage}[c]{0.4\textwidth}
\includegraphics[width=0.55\textwidth]{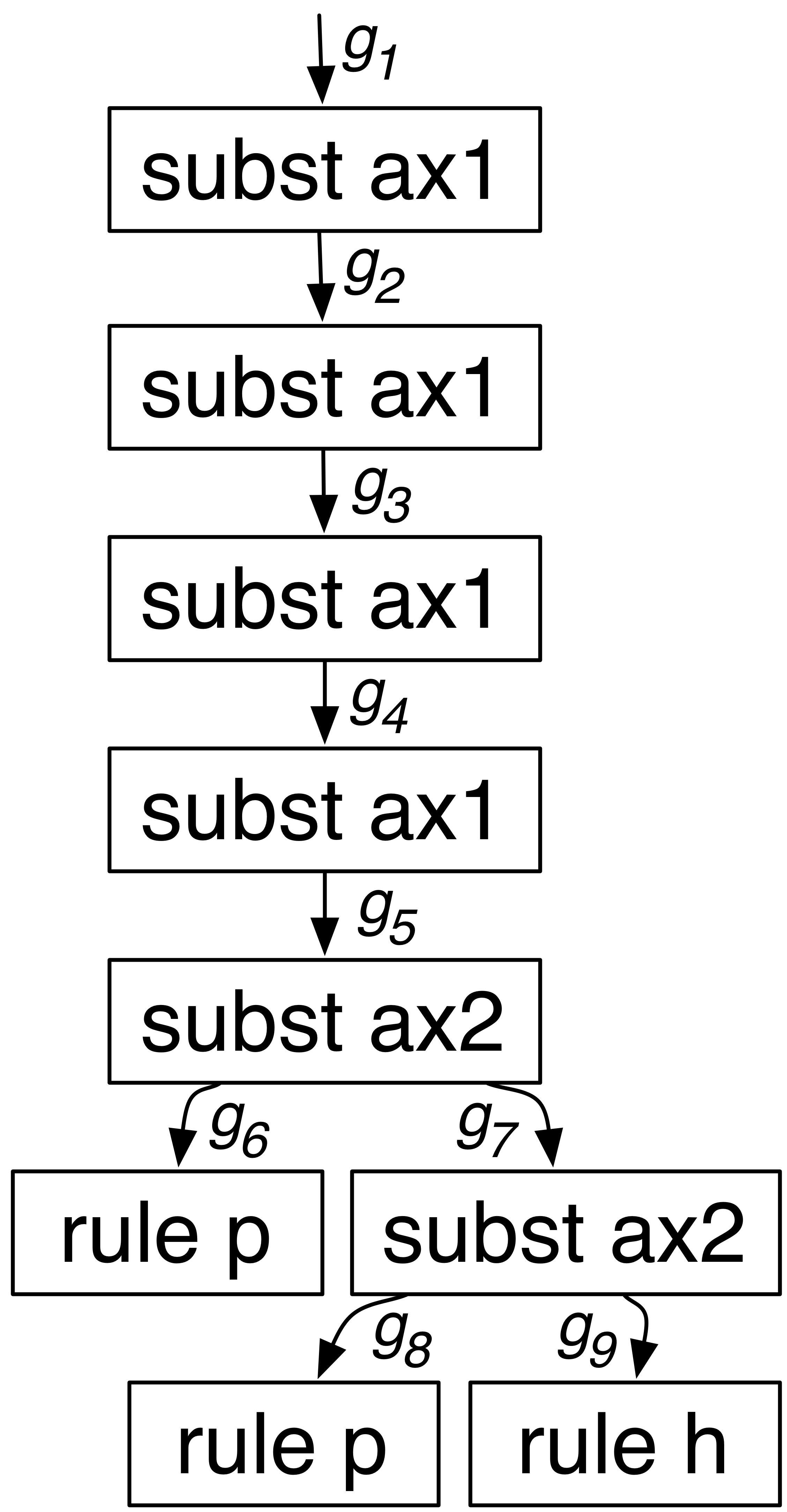}
\end{minipage}
\end{tabular}
\caption{Isabelle proof of (\ref{eq:g1}). The \ncbox{shaded area} illustrate the proof state, and the r.h.s shows the proof tree.}
\label{fig:isaproof}
\end{figure}

Figure \ref{fig:isaproof} illustrates a proof of this conjecture in the Isabelle theorem prover.  Next, consider the following ``similar" conjecture:
\begin{eqnarray}
&p':pure(d),h':a * (((b*c) \wedge d) * e) \vdash ((a * ((b \wedge d) * c)) * e)& \label{eq:g2}
\end{eqnarray}
which again demonstrate the form of a typical proof. This conjecture can be proven by the following sequence of tactic applications:
$$
\textbf{\textsf{apply}}~(\textrm{subst ax1}) ;~ 
\textbf{\textsf{apply}}~(\textrm{subst ax2}) ;~ 
\textbf{\textsf{apply}}~(\textrm{rule p'}) ;~ 
\textbf{\textsf{apply}}~(\textrm{rule h'})
$$
Our goal is to be able to apply some form of \emph{analogous} reasoning to use the proof shown in Figure \ref{fig:isaproof} to
automatically discharge (\ref{eq:g2}). However, a naive reuse of this proof will not work since:
\begin{itemize}
\item There are different number of tactic applications in the two proofs.
\item Naive generalisations such as ``apply \textrm{subst ax1} until it fails" will fail, since \textrm{subst ax1} is still applicable for 
(\ref{eq:g2}) after the first application. Continued application will cause the rest of the proof to fail.
\item The ``analogous" assumptions have different names, e.g. $p$ and $p'$, thus \textrm{rule p} will not work for (\ref{eq:g2}).
\end{itemize}
Even if the proofs are not identical, they are still captured by the same proof strategy. In fact, the proofs can be described as simple version of the the \emph{mutation} proof strategy developed to reason about functional properties in separation logic \cite{Maclean11}.
Here, we assume the existence of a hypothesis $H$ (i.e. $h$ or $h'$),
with some desirable properties we will return to. The strategy can then be described as:

\vspace{-15pt}
\noindent \begin{tabular}{cc}
\begin{minipage}[c]{0.7\textwidth}
\begin{itemize}
\item continue applying \emph{ax1} while there are no symbols
at the same position in the goal and $H$;
\item then apply \emph{ax2} while the goal does not match $H$ and
there is a fact $P$ which can discharge the condition of \emph{ax2}
\item finally, discharge the goal with $H$.
\end{itemize}
\end{minipage}
&
\begin{minipage}[c]{0.25\textwidth}
\begin{figure}[H]
\begin{center}
\includegraphics[width=0.75\textwidth]{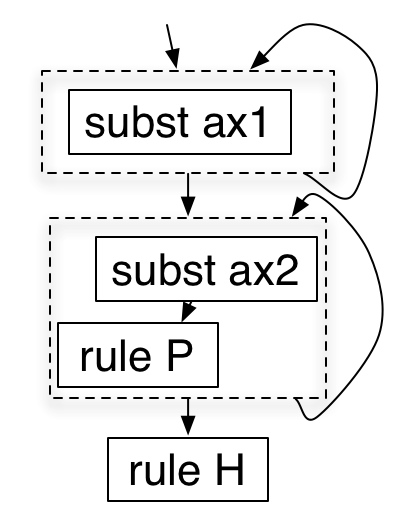}
\end{center}
\vspace{-15pt}
\caption{Mutation}\label{fig:mutation}
\end{figure}
\end{minipage}
\end{tabular} 

The rest of the paper will focus on how we can automatically discover such strategy from the proof shown in Figure \ref{fig:isaproof}. To achieve this a suitable proof strategy representation is required. 
Firstly, as we can see from the strategy, the representation needs to include properties about the \emph{sub-goals/proof states} as well as information about \emph{tactics}. Moreover, sub-goals arising from a tactic application are often treated differently (e.g. the condition arising from the use of \emph{ax2}), thus some ``flow information" is required. From this we argue that
\begin{quote}\it
A graph where the nodes contains the goals, and edges annotated with goal information working as channels for the goals, is a suitable representation to support the automatic generalisation of proof strategies from proofs. 
\end{quote}
Previously, a graph-based language to express proof strategy has  been developed \cite{Grov13}, and we will briefly summarise this in the next section. However, the annotation of goals on edges has not been developed, and developing this is a key challenge in achieving our ambitious goal. A key contribution of this paper is the development of such a {\em goal type} which serves as a specification of a particular goal and can be generalised across proofs. The example has shown that there is vast number of information required which the goal type need to capture:
\begin{itemize}
\item The conclusion to be proven, e.g. $(((c * f) * (d \wedge e)) * b) * a$ initially.
\item The facts available, including local assumptions (such as {\em
    p} and {\em q}), and axioms/lemmas (e.g. \textsf{ax1} and \textsf{ax2}).
\item Properties between facts and the conclusion (or other
  facts). For example,  \textsf{ax2} is applied because the condition of it can be discharged by \textsf{p}.
\item Properties relating goals to tactics; e.g. after applying
  \textsf{ax2} one subgoal is discharged by {\em p} but not the other.
\end{itemize}
Moreover, other information could also be essential for a particular proof strategy, for example: definitions; fixed/shared variables; and variance between steps (e.g. for each step a ``distance" between $h$ and the goal is reduced).

We argue that a language need to be able to capture such properties, and we are not familiar with any proof language which can capture them in a natural way. The development of a \emph{goal type} to capture this is a key contributions and the topic of \S \ref{sec:goaltype}. We then briefly show how this can be utilised when evaluating a conjecture over the strategy in \S \ref{sec:lifting}. The second key contribution of this paper is the topic of \S \ref{sec:generalise}. Here, we utilise the graph language and goal types to generalise a proof into a proof strategy, illustrated by re-discovering the mutation strategy. A key feature here is that we see the goal types as a lattice which can naturally be generalised. This is combined with graph transformations to find common generalisable sub-strategies and loops with termination conditions. We discuss related work and conclude in \S \ref{sec:related} and \S \ref{sec:concl}.



\section{Background on the Proof Strategy Language}\label{sec:background}

The graphical proof strategy language was introduced in \cite{Grov13} built upon the mathematical formalism of 
\textit{string diagrams} \cite{paper:Dixon:10}.
A string diagram consists of \textit{boxes} and \textit{wires}, where the wires are used to connect the boxes. Both boxes and wires can contain data, and data on the edges provides a type-safe mechanism of composing two graphs. Crucially, string diagrams allow \emph{dangling edges}. If such edge has no source, then this becomes and input for the graph, and dually, if an edge as no destination then it is the output of the graph.

In a \emph{proof strategy graph} \cite{Grov13}, the wires are labelled with \emph{goal types}, which is developed in the next section. A box is either a tactic or a list of goals. Such goal boxes
are used for evaluation by propagating them towards the outputs as shown in Figure \ref{fig:eval}.

\begin{figure}[h]
\begin{center}
\scalebox{0.7}{%
\beginpgfgraphicnamed{rewrite_ex}
\begin{tikzpicture}[string graph]
	\begin{pgfonlayer}{nodelayer}
		\node [style=square box] (0) at (-6.25, -0.5) {$t_2$};
		\node [style=none] (1) at (-6.5, -1.5) {};
		\node [style=sg grey vertex] (2) at (-5, 3) {};
		\node [style=none] (3) at (-5.25, 1) {};
		\node [style=none] (4) at (-5, 3.5) {};
		\node [style=none] (5) at (-6, 0) {};
		\node [style=square box] (6) at (-5, 2) {$t_1$};
		\node [style=none] (7) at (-6, -2.5) {};
		\node [style=none] (8) at (-4, 0) {};
		\node [style=none] (9) at (-4, -2.5) {};
		\node [style=none] (10) at (-5, 2.5) {};
		\node [style=none] (11) at (-4.75, 1) {};
		\node [style=none] (12) at (-6, -1) {};
		\node [style=none] (13) at (-7.25, 0.5) {};
		\node [style=none] (14) at (-6.5, 0.5) {};
		\node [style=none] (15) at (-7.25, -1.5) {};
		\node [style=none] (16) at (-3, 0) {$\rewritesto$};
		\node [style=none] (17) at (-6.5, -1) {};
		\node [style=none] (18) at (-5.25, 1.5) {};
		\node [style=none] (19) at (-4.75, 1.5) {};
		\node [style=none] (20) at (-6.5, 0) {};
		\node [style=none] (21) at (1.25, -2.5) {};
		\node [style=none] (22) at (-0.75, -2.5) {};
		\node [style=none] (23) at (0.25, 3.5) {};
		\node [style=none] (24) at (-1.25, 0.5) {};
		\node [style=none] (25) at (0.5, 1) {};
		\node [style=none] (26) at (-1.25, -1) {};
		\node [style=sg grey vertex] (27) at (0.5, 1) {};
		\node [style=none] (28) at (-0.75, 0) {};
		\node [style=none] (29) at (-1.25, 0) {};
		\node [style=none] (30) at (2.25, 0) {$\rewritesto$};
		\node [style=none] (31) at (-2, -1.5) {};
		\node [style=square box] (32) at (-1, -0.5) {$t_2$};
		\node [style=square box] (33) at (0.25, 2) {$t_1$};
		\node [style=none] (34) at (-1.25, -1.5) {};
		\node [style=none] (35) at (0, 1) {};
		\node [style=none] (36) at (0.25, 2.5) {};
		\node [style=none] (37) at (0, 1.5) {};
		\node [style=sg grey vertex] (38) at (0, 1) {};
		\node [style=none] (39) at (0.5, 1.5) {};
		\node [style=none] (40) at (-0.75, -1) {};
		\node [style=none] (41) at (1.25, 0) {};
		\node [style=none] (42) at (-2, 0.5) {};
		\node [style=sg grey vertex] (43) at (6.5, -2) {};
		\node [style=none] (44) at (6.5, -2.5) {};
		\node [style=none] (45) at (4.5, -2.5) {};
		\node [style=none] (46) at (5.5, 3.5) {};
		\node [style=none] (47) at (4, 0.5) {};
		\node [style=none] (48) at (5.75, 1) {};
		\node [style=none] (49) at (4, -1) {};
		\node [style=sg grey vertex] (50) at (4.5, -1.5) {};
		\node [style=none] (51) at (4.5, 0) {};
		\node [style=none] (52) at (4, 0) {};
		\node [style=none] (53) at (7.5, 0) {$\rewritesto$};
		\node [style=none] (54) at (3.25, -1.5) {};
		\node [style=square box] (55) at (4.25, -0.5) {$t_2$};
		\node [style=square box] (56) at (5.5, 2) {$t_1$};
		\node [style=none] (57) at (4, -1.5) {};
		\node [style=none] (58) at (5.25, 1) {};
		\node [style=none] (59) at (5.5, 2.5) {};
		\node [style=none] (60) at (5.25, 1.5) {};
		\node [style=sg grey vertex] (61) at (4, -1.5) {};
		\node [style=none] (62) at (5.75, 1.5) {};
		\node [style=none] (63) at (4.5, -1) {};
		\node [style=none] (64) at (6.5, 0) {};
		\node [style=none] (65) at (3.25, 0.5) {};
		\node [style=sg grey vertex] (66) at (11.75, -2) {};
		\node [style=none] (67) at (11.75, -2.5) {};
		\node [style=none] (68) at (9.75, -2.5) {};
		\node [style=none] (69) at (10.75, 3.5) {};
		\node [style=none] (70) at (9.25, 0.5) {};
		\node [style=none] (71) at (11, 1) {};
		\node [style=none] (72) at (9.25, -1) {};
		\node [style=sg grey vertex] (73) at (9.75, -2) {};
		\node [style=none] (74) at (9.75, 0) {};
		\node [style=none] (75) at (9.25, 0) {};
		\node [style=none] (76) at (12.75, 0) {$\rewritesto$};
		\node [style=none] (77) at (8.5, -1.5) {};
		\node [style=square box] (78) at (9.5, -0.5) {$t_2$};
		\node [style=square box] (79) at (10.75, 2) {$t_1$};
		\node [style=none] (80) at (9.25, -1.5) {};
		\node [style=none] (81) at (10.5, 1) {};
		\node [style=none] (82) at (10.75, 2.5) {};
		\node [style=none] (83) at (10.5, 1.5) {};
		\node [style=sg grey vertex] (84) at (9.25, 0.5) {};
		\node [style=none] (85) at (11, 1.5) {};
		\node [style=none] (86) at (9.75, -1) {};
		\node [style=none] (87) at (11.75, 0) {};
		\node [style=none] (88) at (8.5, 0.5) {};
		\node [style=none] (89) at (13.75, 0.5) {};
		\node [style=none] (90) at (14.5, 0) {};
		\node [style=none] (91) at (15.75, 1.5) {};
		\node [style=none] (92) at (15, 0) {};
		\node [style=sg grey vertex] (93) at (17, -2) {};
		\node [style=sg grey vertex] (94) at (15, -1.5) {};
		\node [style=none] (95) at (17, 0) {};
		\node [style=none] (96) at (16, 2.5) {};
		\node [style=square box] (97) at (16, 2) {$t_1$};
		\node [style=none] (98) at (16.25, 1) {};
		\node [style=none] (99) at (16.25, 1.5) {};
		\node [style=sg grey vertex] (100) at (15, -2) {};
		\node [style=none] (101) at (15.75, 1) {};
		\node [style=square box] (102) at (14.75, -0.5) {$t_2$};
		\node [style=none] (103) at (15, -1) {};
		\node [style=none] (104) at (14.5, -1) {};
		\node [style=none] (105) at (17, -2.5) {};
		\node [style=none] (106) at (15, -2.5) {};
		\node [style=none] (107) at (14.5, -1.5) {};
		\node [style=none] (108) at (16, 3.5) {};
		\node [style=none] (109) at (14.5, 0.5) {};
		\node [style=none] (110) at (13.75, -1.5) {};
	\end{pgfonlayer}
	\begin{pgfonlayer}{edgelayer}
		\draw [style=diredge] (4.center) to (10.center);
		\draw [style=diredge, in=90, out=-90, looseness=1.25] (3.center) to (5.center);
		\draw [style=diredge, in=90, out=-90, looseness=1.00] (12.center) to (7.center);
		\draw [in=90, out=-90, looseness=1.25] (11.center) to (8.center);
		\draw [style=diredge] (8.center) to (9.center);
		\draw [in=-90, out=-90, looseness=2.00] (1.center) to (15.center);
		\draw (15.center) to (13.center);
		\draw [in=90, out=90, looseness=2.00] (13.center) to (14.center);
		\draw (17.center) to (1.center);
		\draw (18.center) to (3.center);
		\draw (19.center) to (11.center);
		\draw [style=diredge] (14.center) to (20.center);
		\draw [style=diredge] (23.center) to (36.center);
		\draw [style=diredge, in=90, out=-90, looseness=1.25] (35.center) to (28.center);
		\draw [style=diredge, in=90, out=-90, looseness=1.00] (40.center) to (22.center);
		\draw [in=90, out=-90, looseness=1.25] (25.center) to (41.center);
		\draw [style=diredge] (41.center) to (21.center);
		\draw [in=-90, out=-90, looseness=2.00] (34.center) to (31.center);
		\draw (31.center) to (42.center);
		\draw [in=90, out=90, looseness=2.00] (42.center) to (24.center);
		\draw (26.center) to (34.center);
		\draw (37.center) to (35.center);
		\draw (39.center) to (25.center);
		\draw [style=diredge] (24.center) to (29.center);
		\draw [style=diredge] (46.center) to (59.center);
		\draw [style=diredge, in=90, out=-90, looseness=1.25] (58.center) to (51.center);
		\draw [style=diredge, in=90, out=-90, looseness=1.00] (63.center) to (45.center);
		\draw [in=90, out=-90, looseness=1.25] (48.center) to (64.center);
		\draw [style=diredge] (64.center) to (44.center);
		\draw [in=-90, out=-90, looseness=2.00] (57.center) to (54.center);
		\draw (54.center) to (65.center);
		\draw [in=90, out=90, looseness=2.00] (65.center) to (47.center);
		\draw (49.center) to (57.center);
		\draw (60.center) to (58.center);
		\draw (62.center) to (48.center);
		\draw [style=diredge] (47.center) to (52.center);
		\draw [style=diredge] (69.center) to (82.center);
		\draw [style=diredge, in=90, out=-90, looseness=1.25] (81.center) to (74.center);
		\draw [style=diredge, in=90, out=-90, looseness=1.00] (86.center) to (68.center);
		\draw [in=90, out=-90, looseness=1.25] (71.center) to (87.center);
		\draw [style=diredge] (87.center) to (67.center);
		\draw [in=-90, out=-90, looseness=2.00] (80.center) to (77.center);
		\draw (77.center) to (88.center);
		\draw [in=90, out=90, looseness=2.00] (88.center) to (70.center);
		\draw (72.center) to (80.center);
		\draw (83.center) to (81.center);
		\draw (85.center) to (71.center);
		\draw [style=diredge] (70.center) to (75.center);
		\draw [style=diredge] (108.center) to (96.center);
		\draw [style=diredge, in=90, out=-90, looseness=1.25] (101.center) to (92.center);
		\draw [style=diredge, in=90, out=-90, looseness=1.00] (103.center) to (106.center);
		\draw [in=90, out=-90, looseness=1.25] (98.center) to (95.center);
		\draw [style=diredge] (95.center) to (105.center);
		\draw [in=-90, out=-90, looseness=2.00] (107.center) to (110.center);
		\draw (110.center) to (89.center);
		\draw [in=90, out=90, looseness=2.00] (89.center) to (109.center);
		\draw (104.center) to (107.center);
		\draw (91.center) to (101.center);
		\draw (99.center) to (98.center);
		\draw [style=diredge] (109.center) to (90.center);
	\end{pgfonlayer}
\end{tikzpicture}}
\endpgfgraphicnamed}
\end{center}
\caption{Evaluation of goals \cite{Grov13}. Goal nodes, represented as circles, are evaluated by propagating them over tactics from inputs to the outputs. An edge works as a channel, and can contain many goals, however, a tactic works on one goal at a time.}
\label{fig:eval}
\end{figure}
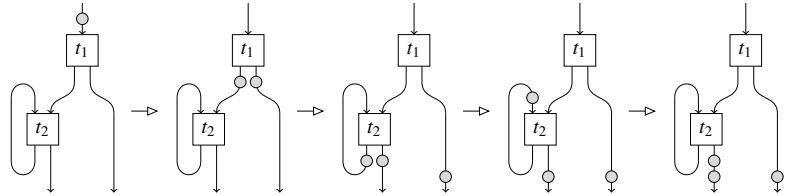

There are two types of tactics. The first type is known as a 
\emph{graph tactic}, which is simple a node holding one or more graphs which be unfolded. This is used to introduce hierarchies to 
enhance readability. A second usage, in the case it holds more then one child graphs, is to represent branching in the search space, as there multiple ways of unfolding such graph. Note however that,
as explained in \cite{Grov13}, graph tactics are evaluated in-place and are thus not unfolded first.

The other type of tactic is an \emph{atomic tactic}. This
corresponds to a tactic of the underlying theorem prover. 
Here, we here assume works on a \textit{proof state} (containing named hypothesis, the open conjecture, fixed variables etc).
When evaluated, such tactic turn a proof state (goal) into a list of new proof states (sub-goals). Since this may also involve search it is returns a set of such list of proof state, thus it has the type
$$
\textit{proof state} \to \{ [\textit{proof state}] \}
$$ 
Here, for a type $\tau$, $[\tau]$ is the type of finite lists of $\tau$ and $\{ \tau \}$ is the type of finite sets whose elements are of type $\tau$. 

For this paper, we assume two atomic tactics: \textsf{subst} $\langle arg \rangle$ and \textsf{rule} $\langle arg \rangle$, which performs a single substitution or resolution step,
respectively. Here, $\langle arg \rangle$ may be both a single rule or a set of rules (all of them are then attempted). It can also be a
\emph{description} of a set of rules, which we call a \emph{class} and is introduced in the next section.

In order to apply an atomic tactic in the strategy language, it has to be \emph{typed} with goal types, also introduced next.
Let $\alpha$ and the $\beta_i$ represent goal type variables. A \emph{typed tactic} is then a function of the form:
$$
\alpha \to \{ [\beta_1] \times [\beta_2] \times \ldots \times [\beta_n] \}
$$
This type has to be reflected in our representation of goal nodes,
which we will return in \S \ref{sec:lifting} after we have developed our notion of \emph{goal types}, which is the next topic.

\section{Towards a Theory of Goal Types}\label{sec:goaltype}





\subsection{Classes} \label{sec:classes}

A goal type must be able to capture the intuition of the user,
potentially using all the information listed in \S
\ref{sec:intro}. This information is then used to guide the proof and
send sub-goals to the correct tactic. To achieve this we firstly need
to capture important properties
of the conclusion of the conjecture. Next, it is important to note that, in general, most of the information available is not relevant, inclusion of it will act as noise (and increase the chance of ``over-fitting" a strategy to a particular proof). 
Thus, we need to be able to separate  
the wheat from the chaff, and capture \emph{properties of the
  `relevant' facts}, where \emph{facts} refer to both lemmas/axioms,
and assumptions which are local to the conclusion. Henceforth we will
term a fact or conclusion an \emph{element}. There are a large set of
such element properties, e.g.: 
\begin{itemize}
\item a particular shape or sub-shape;
\item the symbols used, or symbols at particular positions (e.g. top symbol);
\item certain types of operators are available, e.g. (\ref{eq:g1}) contains associative-commutative operators;
\item the element contains variables we can apply induction to or (shared) meta-variables;
\item certain rules are applicable; 
\item the element's origin, e.g. it is from group theory or it is a property of certain operator.
\end{itemize}
This list is by no means complete, and here we will focus
on two such properties:
\begin{itemize}
\item \textbf{top\_symbol}  describes the top level symbol;
\item \textbf{has\_symbol}  describes the symbols it must contain.
\end{itemize}
Each such \emph{feature} will have data associated:
$$
\textit{data} := \textit{int} ~|~ \textit{term} ~|~ \textit{position}
~| ~\textit{boolean}
$$
where \emph{term} refers to the term of the underlying logic, and 
a \emph{position} refers to an index of a term tree. A \emph{class} describes a family of elements where certain
such features hold. A class, for example, could be a conclusion or a
hypothesis, for which certain properties hold. 
\begin{definition}
A \emph{class} is a map
$$
\textit{class} := \textit{name} \stackrel{m}{\rightarrow} [[\textit{data}]]
$$
such that for each \textit{name} in the domain of a class, there is an associated predicate
on an element, termed the \emph{matcher}.  There are two special cases where the predicates always
succeeds or always fails on certain data, denoted by $\top_f$ and $\bot_f$ as described
below. A class \emph{matches} to a conclusion/fact if the
predicate on each element holds.
\end{definition}
The intuition behind the list of lists of data is that it represent a property in DNF form, e.g.
$[[a,b],[c]]$, which is equivalent to $(a \wedge b) \vee c$. For the conjecture in (\ref{eq:g1}), 
$\{(\textit{top\_symbol} \mapsto [[*]]),(\textit{has\_symbol} \mapsto [[*,\wedge]]) \}$ identifies the
conclusion, while $\{(\textit{has\_symbol} \mapsto [[\textit{pure}]])\}$ identifies the first assumption, 
but not the second, and $\{(\textit{has\_symbol} \mapsto [[\textit{pure}],[*]])\}$ captures both assumptions
and the goal. We call this a {\em semantic representation} of the
data. 

We write the constant space of feature names as $\mathcal{N}$ and, for a class
$C$, with $n \in \mathcal{N}$, $C(n)$ is the data associated with
feature $n$ for class $C$. We define the semantic representation of the
data for a particular feature in a class using the notation
$x^s$ for some data $x$. By semantic representation, we mean that the
structure of the list of data is mapped to a representation about
which we can reason -- for example, above where a list of lists of
data represents a formula in DNF. It is then possible
to reason about this data. For example, for the feature $has\_symbol$
in the conjecture (\ref{eq:g1}) we write for $C(has\_symbol)^s$:
\begin{eqnarray}
&\hspace{-0.3cm}[[a_1 \cdots a_m], \cdots , [b_1 \cdots b_n]]^s =
((\sem{a_1} \cap \cdots \cap \sem{a_m}) \cup \cdots \cup (\sem{b_1}
\cap \cdots \cap \sem{b_n}))& \label{eq:sem}
\end{eqnarray}
\noindent where $\sem{a}$ denotes $a$ as an atom.

Classes form a bounded lattice $(C,\vee,\wedge,\top,\bot)$, on which we can define a meet and a
join. We show how to compute the join ($\wedge$: least upper bound),
and meet ($\vee$: greatest lower
bound) for two classes $C_1$ and $C_2$. We define the most general
class as $\top$ and the empty class as $\bot$. We write the most
general element of $C(f)$ as $\top_f$ and the least general to be $\bot_f$.
\begin{definition}
$C_1 \wedge C_2$ is the greatest lower bound of $C_1$ and $C_2$ if
$\forall n \in \mathcal{N}. (C_1 \wedge C_2)(n) = C_1(n) \wedge_n
C_2(n)$, where $\wedge_n$ computes the greatest lower bound for
feature $n$.
\end{definition}
\begin{definition}
$C_1 \vee C_2$ is the least upper  bound of $C_1$ and $C_2$ if
$\forall n \in \mathcal{N}. (C_1 \vee C_2)(n) = C_1(n) \vee_n
C_2(n)$, where $\vee_n$ computes the least upper bound for
feature $n$.
\end{definition}
For $f = top\_symbol$ or $f = has\_symbol$ we
define $\wedge_f$ and $\vee_f$ as:
\begin{definition} \label{defn:meet} 
$C_1(f) \wedge_f  C_2(f) := \; C_1(f)^s \cap C_2(f)^s$
and $C_1(f) \vee_f  C_2(f) := \; C_1(f)^s \cup C_2(f)^s$
\end{definition}
We further define $\top_f^s$ and $\bot_f^s$ to be $\mathbb{U}$ (the
universal set) and
    $\emptyset$ respectively.  
To show that classes form a partial order, we prove the
following properties about meet and joint:
\begin{theorem}
$\wedge$ and $\vee$ are commutative and associative operations.
\end{theorem}
\begin{proof}
  It suffices to prove that $\wedge_f$ and $\vee_f$ commutative and
  associative for each $f \in \mathcal{N}$. In our example we use
  Definition \ref{defn:meet}. This is
  provable since $\cap$ and $\cup$ are commutative, associative and
  idempotent operations in set theory.
\end{proof}
\begin{theorem}
$\wedge$ and $\vee$ follow the absorption laws $a \vee (a \wedge b) = a$,
and $a \wedge (a \vee b) = a$. 
\end{theorem}
\begin{proof}
It suffices to prove that $\wedge_f$ follow the absorption laws. This
follows from the fact that $\cap$ and $\cup$ are set theoretic
operations. It also follows that $\wedge$ and $\vee$ are idempotent;
$a \wedge a =a$, $a \vee a = a$. 
\end{proof}
Since $\bot$ is $\emptyset$ and $\top$ is $\mathbb{U}$, it is trivial
to show that $C \vee \bot = C$ and $C \wedge \top = C$ for a 
class $C$. Thus, a class form form a bounded lattice.

\textit{Orthogonality} is a key property to reduce non-determinism during
evaluation of a strategy, whilst {\em subtyping} of classes is a key feature for our generalisation techniques discussed in \S \ref{sec:generalise}:
\begin{definition} 
  $C_1$ and $C_2$ are \emph{orthogonal} if $\exists f \in
  \mathcal{N}. C_1(f) \wedge C_2(f) = \bot_f$. We write this as
  $C_1 \bot C_2$. $C_1$ is a \emph{subtype} of $C_2$, written $C_1 ~\subtype~C_2$, if $\forall f \in
\mathcal{N}. \;(C_1(f) \wedge C_2(f)) = C_1(f)$.
\end{definition}
As an example, consider a goal class with features {\em has\_symbol}
and {\em top\_symbol}:
\begin{eqnarray}
C_1:&\{(\textit{top\_symbol} \mapsto [[*]]),(\textit{has\_symbol}
\mapsto [[*,\wedge],[\vee,*]])\}& \nonumber\\
C_2:&\{(\textit{top\_symbol} \mapsto [[\wedge]]),(\textit{has\_symbol}
\mapsto [[*,\wedge],[\vee,*]])\}& \label{eq:c2ex}\\
C_3:&\{(\textit{top\_symbol} \mapsto [[*]]),(\textit{has\_symbol} \mapsto
[[*,\wedge,\vee]])\}& \label{eq:c3ex}
\end{eqnarray}
\noindent $C_2 \bot C_3$ as there is a feature ($top\_symbol$) for
which $C_2(f) \bot_f C_3(f)$, since by the semantics $\sem{\wedge} \cap \sem{*} = \emptyset$.
In order to determine whether $C_1$ is a subtype of $C_3$ we must show
that $(C_1(f) \wedge C_3(f)) = C_3(f)$ for all features. Using
definition \ref{defn:meet} we must prove for $has\_symbol$:
$$
((\sem{*} \cap \sem{\wedge}) \cup (\sem{\vee} \cap \sem{*})) \cap
(\sem{*} \cap \sem{\wedge} \cap \sem{\vee}) = (\sem{*} \cap \sem{\wedge} \cap \sem{\vee})
$$
\noindent which is true and the same for $top\_symbol$ which in this
case follows trivially.

\vspace*{-0.1cm}
\subsection{Links} \label{sec:linkclasses}

A \emph{class} identifies a cluster of elements with certain common properties. However, certain 
types of properties are \emph{between} elements -- e.g.  a conditional
fact can only be applied if the condition can be discharged. Moreover, certain properties rely
on information pertaining to previous nodes in the proof tree, e.g. a measure has to be reduced in a 
rewriting step to ensure termination.
Such properties include;
\begin{itemize}
\item common symbols between two elements, or the position they are at;
\item common shapes between two elements;
\item embedding of one element into another;
\item some form of difference between elements
\item some sort of measure reduces/increases between elements;
\end{itemize}
We call such properties \emph{links}. Moreover, we abstract links to make them relations
between \emph{classes} rather than between elements. Links are given an \emph{existential}
meaning: a link between two classes entails that there exists elements
in them such that a property holds. In addition, we introduce a \emph{parent} function on links to refer to the parent node. 
The meaning of this will become clearer in the next section, where we
discuss evaluation.
\begin{definition}
A \emph{link} is a map
$$
\textit{link} := \textit{name} \times
\textit{class} \times \textit{class} \stackrel{m}{\rightarrow} [[data]] 
$$
such that for each \textit{name} $n$ in the domain of a link, 
there is an associated predicate $n : [[\emph{data}]] \times \emph{element} \times \emph{element} \rightarrow \mathbb{B}$
called a \emph{matcher}. A link \emph{matches} to a conclusion/fact if the
predicate on each element holds.
\end{definition}
We write the constant space of link names as $\mathcal{N}_L$ and for a
link $L$, with $n \in \mathcal{N}_L$, $L(n,C_1,C_2)$ is the data
associated with feature $n$, classes $C_1$ and $C_2$, for link
$L$. 

We will only consider the link features \textbf{is\_match} and \textbf{symb\_at\_pos} for this exposition. The data of the former
are booleans in DNF, and its matcher succeeds if the result of an exact match between the elements is the same as the semantic value
of the data. The data of the latter is lists of position, where for example
$$
\{(\textit{symb\_at\_pos} , C_1 ,  C_2) \mapsto [[pos]]\}
$$
\noindent states that there exists elements of classes $C_1$ and $C_2$ where the symbol at position $pos$ is the same. To state that there is no position
where this is the case, we introduce an element $\bot_f$ for each $f \in \mathcal{N}_L$, as we did with classes.  In general,
there will be more complicated links, with more complicated
output data values. Defining these is ongoing work.

In order to define orthogonality and subtyping we
define the meet and join for each name in $\mathcal{N}_L$.
\begin{definition}
$L_1 \wedge L_2$ is the greatest lower bound of $L_1$ and $L_2$ if
$\forall n \in \mathcal{L}_N. (L_1 \wedge L_2)(n) = L_1(n) \wedge_n
L_2(n)$, where $\wedge_n$ computes the greatest lower bound for link
feature $n$.
\end{definition}
\begin{definition}
$L_1 \vee L_2$ is the least upper  bound of $L_1$ and $L_2$ if
$\forall n \in \mathcal{L}_N. (L_1 \vee L_2)(n) = L_1(n) \vee_n
L_2(n)$, where $\vee_n$ computes the least upper bound for link
feature $n$.
\end{definition} 
As with classes, we introduce a semantic representation for the links
using notation $x^s$ for some data $x$. Since the
data is a list of lists of positions, we use the same semantics as in
(\ref{eq:sem}). The intuition is that we should be able to generalise
the link class to account for the same symbol to  exist at multiple positions
within the hypothesis and conclusions. The proofs and definitions of the lattice theory follow
similarly to  those for classes.

We then define \emph{orthogonality} and \emph{subtyping} for links:
\begin{definition} 
   $L_1$ and $L_2$ are \textit{orthogonal} if $\exists f \in
  \mathcal{L}_N. L_1(f) \wedge L_2(f) = \bot_f$. We write this as
  $L_1 \bot L_2$
$L_1$ is a \emph{subtype} of $L_2$, written $L_1 ~\subtype~L_2$, if $\forall f \in
\mathcal{L}_N. \;L_1(f) \wedge L_2(f) = L_1(f)$.
\end{definition}


\vspace*{-0.4cm}
\subsection{Goal Types} \label{sec:goaltypes}

A goal type is a description of the conclusion, the related facts, 
and the links between them:
\begin{definition}
A goal type is a record: 
{\it \begin{tabbing}  
\qquad GoalType := \=  \{ link : link,  facts : \{ class \} ,  concl : class \}
\end{tabbing}}
\noindent where \emph{concl} is the class describing the conclusion of a goal, \emph{facts} is a set of classes of relevant facts, and \emph{link} is a link relating classes of \emph{facts} and \emph{concl}. 
\end{definition}
Note that we keep a set of \emph{classes} of \emph{facts} to account
for specifying the existence of multiple classes of hypotheses. For example, in the our example conjecture, hypothesis $p$ forms a class $P$ (with \emph{top\_symbol} $pure$), while $h$ forms a class $H$ (with \emph{has\_symbols} $[[\wedge,*]]$). Henceforth we assume that all members of \emph{facts} are orthogonal -- dealing with the general case which
allows overlapping is future work.
Orthogonality and subtyping of two goal types reduces to orthogonality of their respective classes. Due to the assumptions
of orthogonality between the facts, they have an universal interpretation for $\bot$ and an existential interpretation for $\subtype$ 
\begin{tabbing}
$\qquad G_1 ~\bot~ G_2$  \= := \= $G_1(\textit{concl}) ~\bot~ G_2(\textit{concl}) ~\vee~ G_1(\textit{link})
~\bot~ G_2(\textit{link})~ \vee$ \\
\>\> $\forall f_1 \in G_1(\textit{fact}),f_2 \in G_2(\textit{fact}). f_1 ~\bot~ f_2$ \\
$\qquad G_1 ~\subtype~ G_2$ \>:= \>$G_1(\textit{concl})~\subtype~ G_2(\textit{concl}) ~\wedge~G_1(\textit{link})~\subtype~ G_2(\textit{link}) ~\wedge$ \\
\>\> $\exists f_1 \in G_1(\textit{fact}),f_2 \in G_2(\textit{fact}). f_1 ~\subtype~ f_2$
\end{tabbing} 

\section{Lifting of Goals and Tactics}\label{sec:lifting}


Here, we will briefly outline how evaluation is achieved with the goal type introduced. Firstly, recall from Figure \ref{fig:eval} that a single evaluation step is achieved by a tactic by consuming the input goal node on the input and produce the resulting sub-goals on the correct output edges. Since a goal nodes contains list of goals, this can be captured by meta graphical rewrite-rule shown in Figure
\ref{fig:evalrule}. The details are given in \cite{Grov13}, but one evaluation step works as follows:
\begin{center}
\vspace{-12pt}
\noindent \begin{tabular}{cc}
\begin{minipage}{0.45\textwidth}
\begin{enumerate}
  \item Match and partly instantiate the LHS of the meta-rule.
  \item Evaluate the tactic function for the matched input and output types.
  \item Finish instantiating the RHS with the lists $gs_i$ from the tactic.
  \item Apply the fully instantiated rule(s).
\end{enumerate}
\end{minipage}
&
\begin{minipage}{0.52\textwidth}
\vspace{-10pt}
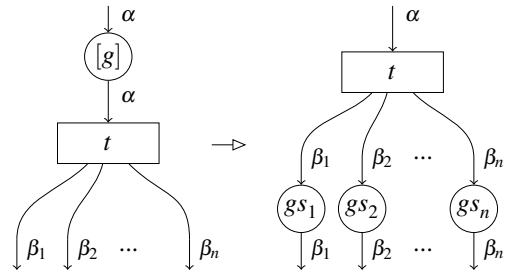
\begin{figure}[H]
\begin{center}
\scalebox{0.9}{%
\beginpgfgraphicnamed{goal_eval}
\begin{tikzpicture}[string graph]
	\begin{pgfonlayer}{nodelayer}
		\node [style=none, wire label] (0) at (-2.5, 3.25) {$\alpha$};
		\node [style=labelled sg vertex, inner sep=1 pt] (1) at (-3, 2.25) {$[g]$};
		\node [style=square box, minimum width=1.5 cm, yshift=0.7 mm] (2) at (-3, 0) {$t$};
		\node [style=none] (3) at (-3, 3.5) {};
		\node [style=none] (4) at (-5.25, -2.25) {};
		\node [style=none] (5) at (-1, -2.25) {};
		\node [style=none] (6) at (-4, -2.25) {};
		\node [style=none, wire label] (7) at (-4.75, -2.5) {$\beta_1$};
		\node [style=none, wire label] (8) at (-3.5, -2.5) {$\beta_2$};
		\node [style=none, wire label] (9) at (-0.5, -2.5) {$\beta_n$};
		\node [style=none] (10) at (-2.5, -2.5) {...};
		\node [style=none] (11) at (0, 0) {$\rewritesto$};
		\node [style=labelled sg vertex] (12) at (6, -1.5) {$\textit{gs}_n$};
		\node [style=labelled sg vertex] (13) at (1.75, -1.5) {$\textit{gs}_1$};
		\node [style=labelled sg vertex] (14) at (3.25, -1.5) {$\textit{gs}_2$};
		\node [style=none] (15) at (3.25, -3) {};
		\node [style=none] (16) at (6, -3) {};
		\node [style=none] (17) at (1.75, -3) {};
		\node [style=none, wire label] (18) at (6.5, -0.25) {$\beta_n$};
		\node [style=none] (19) at (6, -0.25) {};
		\node [style=none] (20) at (4.75, -0.25) {...};
		\node [style=square box, minimum width=1.5 cm, yshift=0.7 mm] (21) at (4, 1.75) {$t$};
		\node [style=none] (22) at (1.75, -0.25) {};
		\node [style=none, wire label] (23) at (3.75, -0.25) {$\beta_2$};
		\node [style=none] (24) at (3.25, -0.25) {};
		\node [style=none, wire label] (25) at (2.25, -0.25) {$\beta_1$};
		\node [style=none, wire label] (26) at (3.75, -2.5) {$\beta_2$};
		\node [style=none, wire label] (27) at (6.5, -2.5) {$\beta_n$};
		\node [style=none, wire label] (28) at (2.25, -2.5) {$\beta_1$};
		\node [style=none] (29) at (4.75, -2.5) {...};
		\node [style=none] (30) at (-4, -3) {};
		\node [style=none] (31) at (-5.25, -3) {};
		\node [style=none] (32) at (-1, -3) {};
		\node [style=none, wire label] (33) at (-2.5, 1.25) {$\alpha$};
		\node [style=none] (34) at (4, 3.5) {};
		\node [style=none, wire label] (35) at (4.5, 3.25) {$\alpha$};
	\end{pgfonlayer}
	\begin{pgfonlayer}{edgelayer}
		\draw [style=diredge] (3.center) to (1);
		\draw [style=diredge, in=90, out=-90] (1) to (2);
		\draw [in=90, out=-135] (2) to (4.center);
		\draw [in=90, out=-45, looseness=1.25] (2) to (5.center);
		\draw [in=90, out=-105] (2) to (6.center);
		\draw [style=diredge] (13) to (17.center);
		\draw [style=diredge] (14) to (15.center);
		\draw [style=diredge] (12) to (16.center);
		\draw [in=90, out=-135] (21) to (22.center);
		\draw [in=90, out=-45] (21) to (19.center);
		\draw [in=90, out=-105] (21) to (24.center);
		\draw [style=diredge] (22.center) to (13);
		\draw [style=diredge] (24.center) to (14);
		\draw [style=diredge] (19.center) to (12);
		\draw [style=diredge] (4.center) to (31.center);
		\draw [style=diredge] (6.center) to (30.center);
		\draw [style=diredge] (5.center) to (32.center);
		\draw [style=diredge] (34.center) to (21);
	\end{pgfonlayer}
\end{tikzpicture}}
\endpgfgraphicnamed}
\end{center}
\caption{Evaluation meta-rule}\label{fig:evalrule}
\end{figure}
\vspace{-10pt}
\end{minipage}
\end{tabular}
\vspace{-7pt}
\end{center}
where $\alpha$ and $\beta_i$ are goal type variables. We assume $t$
is an atomic tactic, but this is trivial to extend to graph tactics. Further note that there are additional rules to split a list into a sequence of singleton lists and delete empty list nodes. For more details we refer to \cite{Grov13}.

In the second step of this algorithm, the underlying tactic has to be lifted from $\textit{proof state} \to \{ [\textit{proof state}] \}$ to the form 
$
\alpha \to \{ [\beta_1] \times [\beta_2] \times \ldots \times [\beta_n] \}.
$

First we need to introduce a \emph{goal}. This can be seen as an instance of a goal type for a particular proof state:
\begin{definition}
A goal is a record:
$$ goal := \{ fmap : \textit{class} \stackrel{m}{\rightarrow} \{ \textit{fact} \},~ \textit{ps} : \textit{proof state},~\textit{parent} : \{\textit{goal}\} \}
$$
\end{definition}
\noindent where \emph{parent} is either a singleton or empty set -- empty if this
is the first goal. Type checking relies on the ``typing predicates" associated with classes and links. A goal $g$ is of type $G$, iff
\begin{itemize}
\item The conclusion in $g(ps)$ matches $G(concl)$.
\item For each class $c \in G(facts)$, $g(fmap)(c)$ is defined, not empty, and each $f \in g(fmap)(c)$ matches $c$.
\item For each $(l,c_1,c_2) \mapsto d \in G(links)$ there exists 
elements $e_1 \in g(fmap)(c_1)$ and $e_1 \in g(fmap)(c_1)$ such that
the $l(d,e_1,e_2)$ holds. Moreover, for each $e_1 \in g(fmap)(c_1)$ there must be an $e_2 \in g(fmap)(c_2)$, such that
$l(d,e_1,e_2)$ (and dually the other way around).
\end{itemize} 

Now, to lift a tactic we need to: \emph{unlift} goal $g$ to project the underlying proof state; apply the tactic; and lift the resulting proof states to goals of a type in $\{\beta_1,\ldots,\beta_n\}$ (which becomes instantiated to specific goal types when matching the
RHS in the first step). Then, for a list $L$ of proof states, let $lp(\beta_1,\ldots,\beta_n; L)$ be the set of all partitions of $L$ 
lifted into $n$ lists of goals $\{ (\textit{map lift }L_1)$, $\ldots,$ $(\textit{map lift } L_n) \}$, such that all of the goals in 
the $i$-th list have goal type $\beta_i$. Then, we define lifting as:
$$
  \textit{lift}(tac) = \lambda g.~
  \left\{
   \begin{array}{lll}
    lp(\beta_1,\ldots,\beta_n; \textit{tac}(\emph{unlift}(g))) & \textrm{ if }g\textrm{ is of type }\alpha \\
    \emptyset & \textrm{ otherwise}
  \end{array}\right.
$$ 
We are then left to define \emph{unlifting} and \emph{lifting} for a single goal node and a single goal type. Firstly, a naive \emph{unlifting} of a goal simply projects the goal state. 
More elaborate unliftings are tactic dependent, and may e.g. add all facts from a particular fact class as active assumptions beforehand. 

\emph{Lifting} is a partial function, and an element of \emph{lp} is only defined if lifting of all elements succeeds. There are several
(type-safe) ways to implement \emph{lifting}. Here, we show a 
procedure which assumes that all relevant information is passed 
down the graph from the original goal node. Any fact ``added" to
a goal node is thus a fact generated by the tactic. However, one
may ``activate" existing facts explicitly in the tactic which will then be used by lifting. A new goal $g'$ is then lifted as 
follows, using the (new) proof state $ps'$, previous goal $g$,
and goal type $G$ as follows:
\begin{enumerate}
\item Set fields $g'(parent)$ to $g$, and $g'(ps)$ to $ps'$, fail if
  the conclusion does not match $G(concl)$.
\item For each $c \in G(facts)$, set $g'(facts)(c)$ to be all
facts in the range of $g(facts)$ and newly generated facts which matches $c$. If for any $c \in G(facts)$, $g'(facts)(c)$ is empty (or undefined) then fail. 
\item Check all link features. For each $c \in G(facts)$ which is used by a link feature, filter out any element $e \in g'(facts)(c)$
not ``captured" by a link related link match. Fail if there does not
exist an element in the related classes which holds for any of the links or any $g'(facts)(c)$ (for $c \in G(facts)$) is empty after this filtering step. 
\end{enumerate}
\vspace*{-0.2cm}

\section{Generalising Strategies}\label{sec:generalise}

A proof is generalised into a strategy by first lifting the proof tree into a proof strategy graph, and then apply graph transformation techniques which utilises the goal type lattice 
to generalise goal types. Simple generalisation of tactics are also used. One important property when performing such generalisations, is that \emph{any valid proofs on a strategy should also be valid after}, which we will provide informal justification for below. However, note that we do not deal with termination.

\begin{figure}[t]
  \begin{center}
    \includegraphics[width=0.85\textwidth]{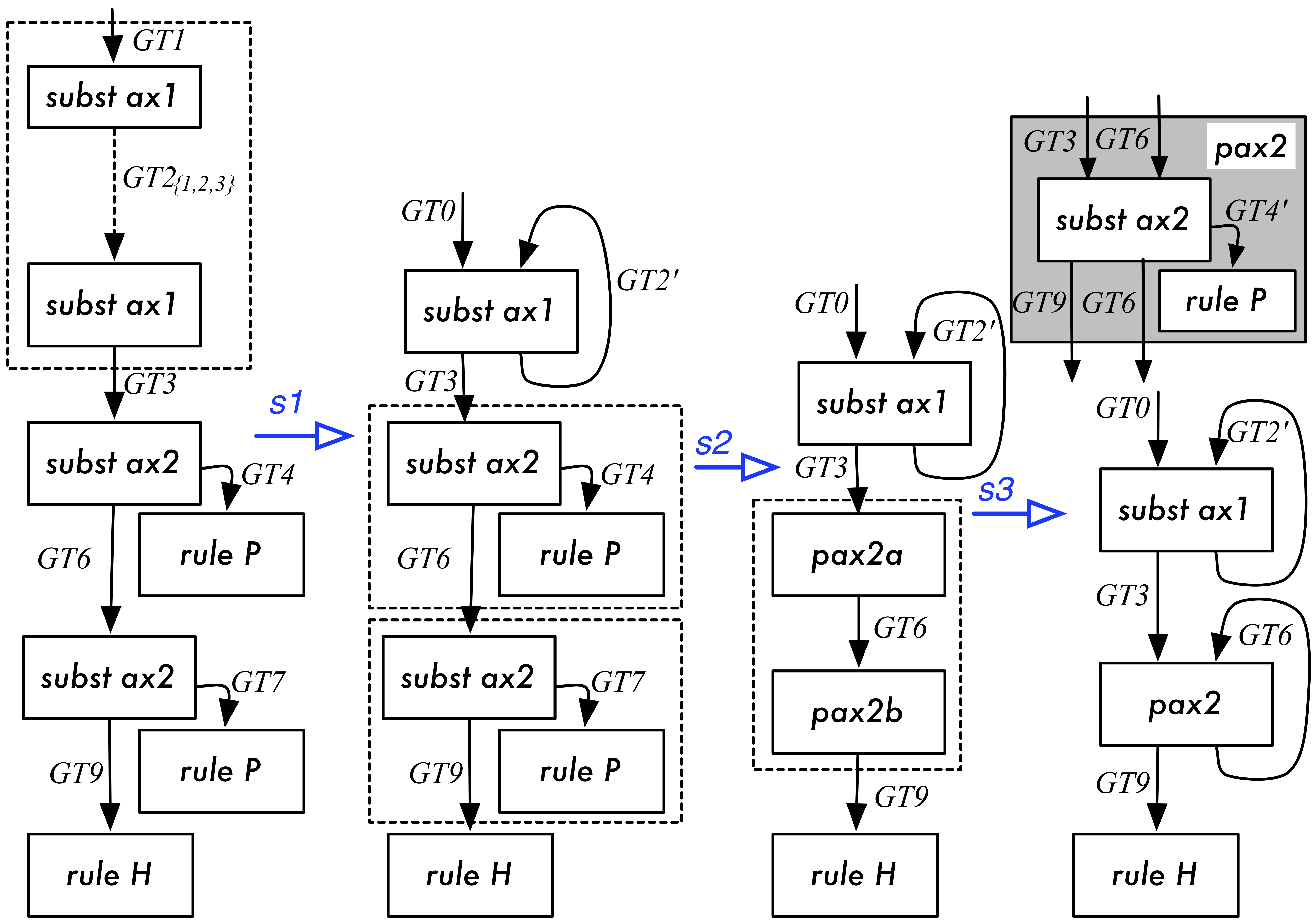}
  \end{center}
  \vspace{-15pt}
   \caption{Steps made to generalise proof of (\ref{eq:g1})}
   \label{fig:genseq}
\end{figure}

\subsection{Deriving Goal Types from Proof States}

In this section we will discuss how to generalise the proof 
shown in Figure \ref{fig:isaproof} into the mutation strategy shown in Figure \ref{fig:mutation}, utilising the lattice structure of goal types. 

However, first we need to turn the proof tree of Figure \ref{fig:isaproof} into a low-level proof strategy graph of the same shape. We utilise techniques described in \cite{Velykis12} to get the initial proof tree. Now, since the shape is the same this reduces to (1) generalising proof states into goal types and (2) generalising the tactics. 

(1) To generalise the proof state into goal type we have taken an approach which can be seen as a ``locally maximum" derivation of goal type, where each assumption becomes a separate class, and make each class as specific as possible. Any link features that holds are also included. Consequently, the goal type will be as far down the lattice as possible whilst still being able to lift the goal state it is derived from. To illustrate, we will show how the proof state (\ref{eq:g1}) is lifted to goal type $GT1$. Let 
$$
\begin{array}{rlcrl}
H & = \{has\_symbol \mapsto [[*,\wedge]]\},\{top\_symbol \mapsto [[*]]\} \\
P & = \{has\_symbol \mapsto [[pure]],top\_symbol \mapsto [[pure]]\} \\
G & = \{has\_symbol \mapsto [[*,\wedge]],top\_symbol \mapsto [[*]]\} \\
L & = \{(symb\_at\_pos,G,H) \mapsto [[\bot]],
(symb\_at\_pos,G,P) \mapsto [[\bot]],
(symb\_at\_pos,H,P) \mapsto [[\bot]]\}.
\end{array}
$$
Then $GT1$ becomes $\{link : L, facts : \{H,P\},concl : G \}$. Note that the last two link features are useless, and are therefore ignored henceforth. However, this shows that in the presence of larger goal states and/or more properties heuristics will be required to reduce the size of the goal types, and filter out such ``useless information". This is future work.

(2) Tactics are kept with the difference that if a local assumption is used (e.g. $h$ or $p$) their respective class is used instead.

The resulting tree is shown left-most of Figure \ref{fig:genseq}. For
space reasons we have not included the goal types, but provided
a name when referred to in the text. This is slightly more general than the original proof as it allows a very slight variation of the goals. However, it still e.g. relies on the exact number of application of each tactic. 

\subsection{Generalising Tactics}

Next, we need to generalise tactics,  A simple example of this is when sets of rules are used as arguments
for the \textsf{subst} and \textsf{rule} tactics.
Here, $\textsf{subst}~R_1$ and $\textsf{subst}~R_2$ can be generalised into $\textsf{subst}~(R_1 \cup R_2)$. Another example turns a tactic into a graph tactic which nest both these tactics (and can be unfolded to either). A proviso for both is that their input and output goal types can be generalised. Both these generalisations only increases the search space and are thus proper generalisations. 

Graph tactics can also be generalised by generalising the graph they nest into one. We return to this with an example below. We will use the notation $\emph{gen}(t_1,t_2)$ for the generalisation of the two given tactics.

\subsection{Generalising Goal Types}

In the context of goal types: {\em generalisation} 
refers to computing the most general goal type for two existing goal
types; while {\em weakening} applies to only one goal type and
makes the description of it more general.
Crucial to both generalisation and weakening is that multiple
possible generalised and weakened goal types exist. 

We use the notion of a {\em least upper bound} for a goal type lattice, 
described in \S \ref{sec:goaltype} using the join operator $\vee$, to define
generalisation for goal types. 
For a class $C$, we write:
\begin{definition}
$C$ is a generalisation of $C_1$ and $C_2$, also written
$C=gen(C_1,C_2)$, if $\forall f \in \mathcal{N}. \;C(f) = C_1(f) \vee C_2(f)$.
\end{definition}
\noindent As an example, consider the two classes shown in (\ref{eq:c2ex})
and (\ref{eq:c3ex}). We can compute $G = gen(C_1,C_2)$ by appealing to the set theoretic
semantics and tranferring back to the class representation. For $f_1 =
top\_symbol$ and $f_2 = has\_symbol$ we compute
$$
\begin{array}{rll}
C(f_1)^s = & (\sem{\wedge} \cup \sem{*}) & \leadsto C(f_1) = [[\wedge],[*]] \\
C(f_2)^s = & 
\multicolumn{2}{l}{((\sem{*} \cap \sem{\wedge}) \cup (\sem{\vee} \cap
\sem{*})) \cup (\sem{*} \cap \sem{\wedge} \cap \sem{\vee})} \\
= & ((\sem{*} \cap \sem{\wedge}) \cup (\sem{\vee} \cap \sem{*}))
\qquad\qquad\qquad & \leadsto C(f_2) = [[*,\wedge],[\vee,*]]
\end{array}
$$
\noindent producing a generalised
class:
\begin{eqnarray*}
C:&\{(\textit{top\_symbol} \mapsto [[\vee],[*]]),(\textit{has\_symbol} \mapsto
[[*,\wedge],[\wedge,\vee]])\}&
\end{eqnarray*}
The definition of generalisation for links  extends similarly from
its associated lattice theory described in
\S\ref{sec:linkclasses}. Recall that we assume orthogonality of fact
classes. We define a function {\em gen\_map} over two sets of (fact) classes, which generalises pairwise each fact class. Here, for any two fact classes $H_1$ and $H_2$ in the generalised set of fact classes, where $(H_1 ~\subtype ~H_2) \;\wedge\; \neg(H_1 ~\bot~
H_2)$ we only retain  $H_2$, thus
ensuring orthogonality. We can then define a function $gen$ on goal types to be
{\it \vspace{-5pt}
\begin{tabbing}
$\qquad$ gen($G_1$, $G_2$) := \{ \= \textit{concl} \= = gen($G_1$(\textit{concl}),$G_2$(\textit{concl})), \\
\> \textit{facts} \> = gen\_map($G_1$(\textit{facts}),$G_2$(\textit{facts})),\\
\> \textit{link} \> =  gen($G_1$(\textit{link}),$G_2$(\textit{link}))
\end{tabbing}}
\vspace{-10pt}

\begin{figure}[t]
\vspace{-10pt}
\includegraphics[width=\textwidth]{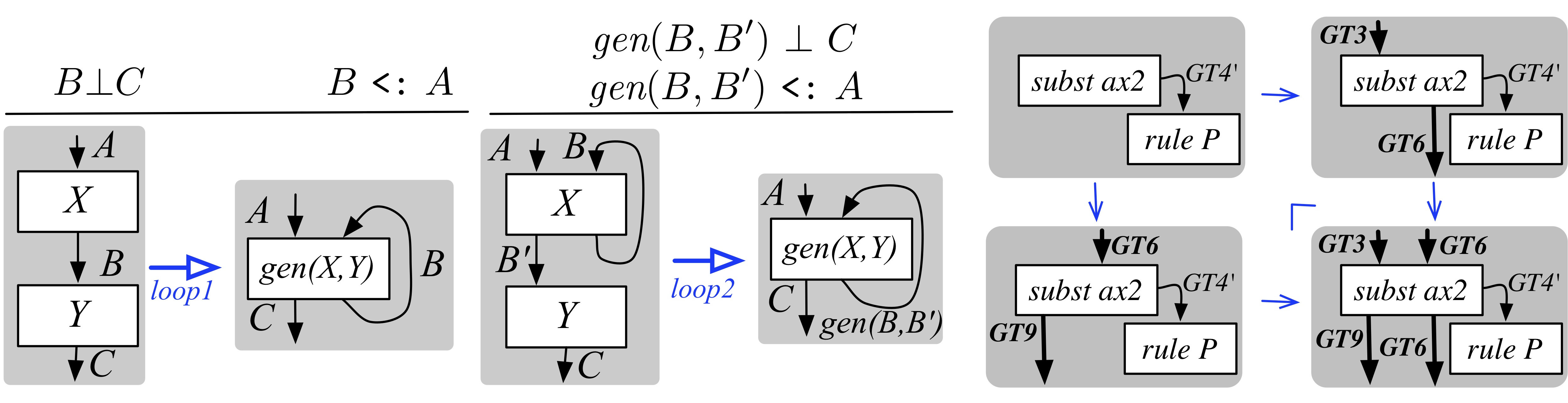}
\vspace{-20pt}
\caption{Generalisations: left/middle: loops rules; right: push-out over common subgraph}
\label{fig:genrules}
\end{figure}

\subsection{(Re-)Discovering the Mutation Strategy} \label{sec:stor}

Armoured with the techniques for generalising the edges and nodes of a proof strategy, we now develop two techniques which allows us to generate our proof into the required mutation strategy.

Firstly, we need to abstract over the number of repeated sequential applications of the same tactic -- i.e. we need to discover loops.  When working in a standard LCF tactic
language \cite{Gordon79}, the problem is
to know: (a) on which goals (in the case of side conditions) the tactic should be repeated,
and (b) when to stop. This was highlighted in  \cite{phd:Duncan:02}, where a regular expression language, closely aligned with common LCF tacticals, was used to learn proof tactics, and hand-crafted heuristics were defined to state when to stop a loop (which by the way would fail for our example).

The advantage of our approach, is that we can utilise the goal types to identify termination 
conditions -- reducing termination and goal focus to the same case, thus also handling the more general proof-by-cases paradigm.  
We illustrate our approach with what can be seen as an inductive representation of tactic looping, as shown by rules \emph{loop1} and \emph{loop2} of Figure \ref{fig:genrules}.
For \emph{loop1}, we can see that it is correct
since $B \bot C$ ensures that a goal will exit the loop when it matches
$C$. Moreover, the $B~\subtype~A$ pre-condition ensures that the tactic
can handle the input type. For \emph{loop2}, similar arguments holds for the generalised $\textit{gen}(B,B')$ edge.

Consider the left most graph of Figure \ref{fig:genseq}, which is the proof tree lifted to a graph. 
Here, the stippled box highlights the sub-graph which matches with the rules shown above. \emph{loop1} is applied first, followed by two applications
of \emph{loop2}. The classes are identical so we only discuss
link classes, which have the following values:
\begin{eqnarray*}
GT1(link) = GT2_n(link) =&\{(symb\_at\_pos,G,H) \mapsto&[[\bot]]\}\\
GT3(link) =&\{(symb\_at\_pos,G,H) \mapsto&[[1]]\}
\end{eqnarray*}
where $GT2_n$ denote the goal types in the intermediate stages of the repeated
application of tactic {\em subst ax1}. Now, for the sequence to be detected as a loop, we must
first discover
$$
GT2' = gen(gen(GT2_1,GT2_2),gen(GT2_3) = GT2_1
$$
\noindent and show $GT2'~\subtype~GT0$ and $GT2' \bot GT3$. These are
both true since $GT2'$ and $GT1$ are equal, and $GT2'$ and $GT3$
are orthogonal due to the existence of $\bot$ in the data argument
denoting an empty feature.

The next step (s2) of Figure \ref{fig:genseq} \emph{layers} the highlighted sub-graphs into the graph tactics \textit{pax2a} and \textit{pax2b}. Such layering can be done for a (connected) sub-graph if the inputs and outputs of the sub-graphs are respectively orthogonal. 

Next, we again apply rule \emph{loop1} to the \textit{pax2a} and \textit{pax2b} sequence. However, this requires us to generalise these two graph tactics, i.e. combining the two graphs they contain into one. Now, as shown in \cite{paper:Dixon:10}, in the category of string graphs, two graphs are composed by a \emph{push-out} over a common boundary. We can combine two graph tactics in the same way by a push-out over the \emph{largest common sub-graph}. This is shown on the right-most diagram of Figure \ref{fig:genrules}, which becomes the last step (s3) of
Figure \ref{fig:genseq}.

This graph is in fact the mutation strategy of Figure \ref{fig:mutation}, with the addition that we have given semantics to the edges. Now, the first feedback loop is
identified by $\{(symb\_at\_pos,concl,H) \mapsto [[\bot]]\}$,
while the second feedback loop is identified by $\{(is\_match,concl,H) \mapsto [[false]]\}$.



%

\section{Related Work}\label{sec:related}

We extend \cite{Grov13}, which introduces the underlying strategy
language, by developing a theory for goal types which we show form
a lattice, and using this property to develop techniques for generalising strategies. 

Our goal types can be seen as a lightweight implementation of pre/post-condition used in \emph{proof planning}~\cite{Bundy88} -- with the  additional property that the language  captures the flow of goals. It can be seen as further extending the marriage of \emph{procedural}
and \emph{declarative} approaches to proof strategies \cite{Autexier10,Harrison96,Giero:07}, and addressing issues related
to goal flow and goal focus highlighted in \cite{Asperti09} -- for a more
detailed comparison we refer to \cite{Grov13}.

The lattice based techniques developed for goal type generalisation
is similar to \emph{antiunification} \cite{Plotkin69} which generalises
two terms into one (with substitutions back to the original terms). Whilst
each feature is primitive, the goal type has several dimensions. More
expressive class/link features, which is future work, may require
higher-order anti-unification \cite{Krumnack07} -- and such ideas may
also be applicable to graph generalisations. Other work that may become
relevant for our techniques are graph abstractions/transformations
used in algorithmic heap-based program verification techniques, 
such as \cite{Boneva07}, and for parallelisation of functional 
programs \cite{Grov10}.

As already discussed, the problem when ignoring goal information, 
is that one cannot describe e.g. where to send a goal or
when to terminate a loop, in a way sufficiently abstract to capture
a large class of proofs. Instead, often crude, heuristics have to be used
in the underlying tactic language. This is the case for 
\cite{phd:Duncan:02}, which uses a regular expression language (close to LCF tactics), originally developed in \cite{misc:Jamnik:01} to learn
proof plans. \cite{misc:Jamnik:01} further claims that \emph{explanation based generalisation} (EBG)\cite{Mitchell97} is applied to derive pre/post-conditions, but no details of this are provided. An EBG approach is also applied to generalise Isabelle proof terms
into more generic theorems in \cite{Lueth04}. This could provide an alternative
starting point for us, however, one may argue that much of the user intent will be lost by working in the low-level proof term representation. Further, note that our work focuses on proof of conjectures which requires
\emph{structure}, meaning machine learning techniques -- such as 
\cite{TsivtsivadzeUGH11}, which learns heuristics to select relevant axioms/rules for automated provers -- are not sufficient. However, in
 \cite{Grov12}, an approach to combine essentially our techniques, with  more probabilistic techniques to cluster interactive proofs  \cite{Komendantskaya13}, was outlined.

We would also like to utilise work on 
proof and proof script \emph{refactoring} \cite{paper:Whiteside:11}. This could be achieved either 
as a pre-processing step, or by porting these techniques
to our graph based language. Finally, albeit for source code, \cite{Mens05} argues for the use of graphs to perform refactorings, which further justifies our graph based representation of proof strategies for the work presented here.

\section{Conclusion and Future Work}\label{sec:concl}

In this paper we have reported on our initial results in creating a technique to generalise proof into high-level proof strategies which can be used to automatically discharged similar conjectures. This paper has two contributions: (1) the introduction of goal type to describe properties of goals using a lattice structure to enable generalisations; (2) two generic techniques, based upon loop discovery to generalise a proof strategy. The techniques was motivated and illustrated by an example from separation logic. We are in the process of implementation in Isabelle combined with the Quantomatic graph rewriting engine \cite{Quantomatic}. Next plan to implement these methods in order to test them on more examples, using a larger set of properties to represent the goal types. In particular, we are interested in less syntactic properties, such as the origin of a goal,
or if it is in a decidable sub-logic.

We also showed how the lattice structure corresponds to sub-typing, and we plan to incorporate sub-typing in the underlying theory of the language in order to utilise it when composing graphs. Further, we plan to develop more techniques for generalising graphs, which may include develop an underlying theory of \emph{graph
generalisation}, which will be less restrictive than rewriting.

Finally, we have already touched upon the need for heuristic guidance in this work, as there will be many ways of generalising. We are also planning to apply the techniques to extract strategies from a corpus of proofs. Here we believe we have a much better chance of finding and
generalising common \emph{sub-strategies}, and may also incorporate
probabilistic techniques as a pre-filter \cite{Grov12}. Such work may
help to indicate which class/link features are more common, and can be used to improve the generalisation heuristics discussed above. 
Further, we would like to remove the restriction that facts have to be
orthogonal, and improve the sub-typing to handle this case.

We only briefly discussed the process of turning proofs into initial low-level proof strategy graphs. With partners on the AI4FM project (\url{www.ai4fm.org}) we are working on utilising their work
on  capturing the full \emph{proof process}, where the user
may (interactively) highlight the key features of a proof (step)
\cite{Velykis12}. This can further help the generalisation heuristics.


\subsection*{Acknowledgements}
This work has been supported by EPSRC grants: EP/H023852, EP/H024204 and EP/J001058. We would like to
thank Alan Bundy, Aleks Kissinger, Lucas Dixon, members of the
AI4FM project, Katya Komendantskaya, Jonathan Heras and Colin Farquhar
for feedback and discussions.

\bibliographystyle{plain}
\bibliography{stratlang}

\begin{thebibliography}{10}

\bibitem{Asperti09}
A.~Asperti, W.~Ricciotti, C.~Sacerdoti, and C.~Tassi.
\newblock A new type for tactics.
\newblock In {\em PLMMS'09}, pages 229--232, 2009.

\bibitem{Autexier10}
Serge Autexier and Dominik Dietrich.
\newblock A tactic language for declarative proofs.
\newblock In {\em ITP'10}, volume 6172 of {\em LNCS}, pages 99--114. Springer,
  7 2010.

\bibitem{Boneva07}
I.B. {Boneva}, A.~{Rensink}, M.E. {Kurban}, and J.~{Bauer}.
\newblock Graph abstraction and abstract graph transformation.
\newblock Technical Report TR-CTI, University of Twente, July 2007.

\bibitem{Bundy88}
A.~Bundy.
\newblock The use of explicit plans to guide inductive proofs.
\newblock In R.~Lusk and R.~Overbeek, editors, {\em CADE9}, pages 111--120.
  Springer-Verlag, 1988.

\bibitem{AI4FM:avocs09}
A.~Bundy, G.~Grov, and C.B. Jones.
\newblock Learning from experts to aid the automation of proof search.
\newblock In {\em PreProc of AVoCS'09}, pages 229--232, 2009.

\bibitem{paper:Dixon:10}
Lucas Dixon and Aleks Kissinger.
\newblock Open graphs and monoidal theories.
\newblock {\em CoRR}, abs/1011.4114, 2010.

\bibitem{phd:Duncan:02}
Hazel Duncan.
\newblock {\em The use of Data-Mining for the Automatic Formation of Tactics}.
\newblock PhD thesis, University of Edinburgh, 2002.

\bibitem{Giero:07}
M.~Giero, F.~Wiedijk, and Mariusz Giero.
\newblock {MM}ode, a mizar mode for the proof assistant coq.
\newblock Technical report, January~07 2004.

\bibitem{Gordon79}
Michael J.~C. Gordon, Robin Milner, and Christopher~P. Wadsworth.
\newblock {\em Edinburgh LCF}, volume~78 of {\em Lecture Notes in Computer
  Science}.
\newblock Springer, 1979.

\bibitem{Grov13}
G.~Grov and A.~Kissinger.
\newblock A graphical language for proof strategies.
\newblock Submitted to ITP'13. Also available at arXiv:1302.6890.

\bibitem{Grov12}
G.~Grov, E.~Komendantskaya, and A.~Bundy.
\newblock A statistical relational learning challenge - extracting proof
  strategies from exemplar proofs.

\bibitem{Grov10}
G.~Grov and G.~Michaelson.
\newblock Hume box calculus: robust system development through software
  transformation.
\newblock {\em HOSC}, 23:191--226, 2010.

\bibitem{Harrison96}
John Harrison.
\newblock A mizar mode for {HOL}.
\newblock In {\em TPHOLs}, volume 1125 of {\em LNCS}, pages 203--220. Springer,
  1996.

\bibitem{misc:Jamnik:01}
M.~Jamnik, M.~Kerber, and C.~E Benzmuller.
\newblock {L}earning {M}ethod {O}utlines in {P}roof {P}lanning.
\newblock Technical Report CSRP-01-8, University of Birmingham (CS), 2001.

\bibitem{Lueth04}
E.B. Johnsen and C.~L\"uth.
\newblock Theorem reuse by proof term transformation.
\newblock In {\em {TPHOLs 2004}}, volume 3223 of {\em LNCS}, pages 152--167.
  Springer, 2004.

\bibitem{Quantomatic}
A.~Kissinger, A.~Merry, L.~Dixon, R.~Duncan, M.~Soloviev, and B.~Frot.
\newblock Quantomatic, 2011.

\bibitem{Komendantskaya13}
E.~Komendantskaya, J.~Heras, and G.~Grov.
\newblock Machine learning in proof general: Interfacing interfaces.
\newblock {\em CoRR}, abs/1212.3618, 2012.

\bibitem{Krumnack07}
U.~Krumnack, A.~Schwering, H.~Gust, and K-U K{\"u}hnberger.
\newblock Restricted higher-order anti-unification for analogy making.
\newblock In {\em AJAI 2007}, volume 4830 of {\em LNAI}, pages 273--282.
  Springer, 2007.

\bibitem{Maclean11}
Ewen Maclean and Andrew Ireland.
\newblock Mutation in linked data structures.
\newblock In {\em ICFEM}, volume 6991 of {\em LNCS}, pages 275--290. Springer,
  2011.

\bibitem{Mens05}
T.~Mens, N.~Van Eetvelde, S.~Demeyer, and D.~Janssens.
\newblock Formalizing refactorings with graph transformations.
\newblock {\em Journal of Software Maintenance}, 17(4):247--276, 2005.

\bibitem{Mitchell97}
Tom Mitchell.
\newblock {\em Machine Learning}.
\newblock McGraw-Hill, 1997.

\bibitem{Plotkin69}
G.~D. Plotkin.
\newblock A note on inductive generalization.
\newblock In {\em Machine Intelligence 5}, pages 153--163, Edinburgh, 1969.
  Edinburgh University Press.

\bibitem{SepLogic02}
J.~C. Reynolds.
\newblock Separation logic: A logic for shared mutable data structures.
\newblock In {\em Logic in Computer Science}, pages 55--74. IEEE Computer
  Society, 2002.

\bibitem{TsivtsivadzeUGH11}
E.~Tsivtsivadze, J.~Urban, H.~Geuvers, and T.~Heskes.
\newblock Semantic graph kernels for automated reasoning.
\newblock In {\em Proc. 11th SIAM Int. Conf. on Data Mining}, pages 795--803.
  SIAM / Omnipress, 2011.

\bibitem{Velykis12}
Andrius Velykis.
\newblock Inferring the proof process.
\newblock In Christine Choppy, David Delayahe, and Ka\"{i}s Kla\"{i}, editors,
  {\em FM2012 Doctoral Symposium}, Paris, France, August 2012.

\bibitem{paper:Whiteside:11}
I.~Whiteside, D.~Aspinall, L.~Dixon, and G.~Grov.
\newblock Towards formal proof script refactoring.
\newblock In {\em CICM'11}, volume 6824 of {\em LNCS}, pages 260--275.
  Springer, 2011.

\end{thebibliography}



\end{document}